\documentclass[letterpaper,12pt]{article}

\usepackage{amsthm}
\usepackage{amsfonts}
\usepackage{amsmath}
\usepackage{graphicx}
\usepackage{lineno}
\usepackage{setspace} 
\usepackage{verbatim}

\newtheorem{theorem}{Theorem}
\newtheorem*{theorem*}{Theorem}
\newtheorem{conjecture}[theorem]{Conjecture}
\newtheorem{corollary}[theorem]{Corollary}

\newtheorem{lemma}[theorem]{Lemma}

\newtheorem{observation}{Observation}

\title{Asymptotics of the chromatic number for quasi-line graphs}
\author{Andrew D.\ King\thanks{Corresponding author: {\tt andrew.d.king@gmail.com}, IEOR Department, Columbia University, New York.  Research supported by NSERC doctoral and postdoctoral fellowships and the European project {\sc ist fet Aeolus}, contract number IP-FP6-015964.}\ \ and Bruce Reed\thanks{Research supported in part by a Canada Research Chair.}}

\begin{document}

\maketitle

\begin{abstract}

As proved by Kahn, the chromatic number and fractional chromatic number of a line graph agree asymptotically.  That is, for any line graph $G$ we have $\chi(G) \leq (1+o(1))\chi_f(G)$.  We extend this result to quasi-line graphs, an important subclass of claw-free graphs.  Furthermore we prove that we can construct a colouring that achieves this bound in polynomial time, giving us an asymptotic approximation algorithm for the chromatic number of quasi-line graphs.
\end{abstract}


\section{Introduction}

The chromatic number $\chi$ of a graph $G=(V,E)$ is notoriously difficult to compute or even approximate.  Bellare, Goldreich, and Sudan proved that it is NP-hard to approximate $\chi$ to within a factor of $|V|^{1/7}$ \cite{bellaregs98}.  Like any NP-complete problem, we can formulate $\chi$ as the solution to an integer program, and may hope to approximate it by solving the fractional LP relaxation.  In this case, we seek a {\em fractional vertex $c$-colouring} of a graph $G$, i.e.\ a nonnegative weighting $w$ on the stable sets of $G$ such that $\sum_{S}w(S) \leq c$, and for every vertex $v$, $\sum_{S \ni v} w(S) = 1$.  The {\em fractional chromatic number} of $G$, written $\chi_f(G)$, is the smallest $c$ for which $G$ has a fractional vertex $c$-colouring.  In general the approach of approximating $\chi$ by finding $\chi_f$ is impractical for two reasons.  First, $\chi$ is not necessarily bounded above by any function of $\chi_f$ (see for example \cite{kingthesis} \S3.1).  Second, Lov\'asz proved that $\chi \leq \lceil \log(|V|)\chi_f\rceil$ \cite{lovasz75}, so in general it is NP-hard to approximate $\chi_f$ to within a factor of $|V|^{\epsilon} $ for any $\epsilon<1/7$.

In spite of these difficulties, the fractional approach is effective for certain restricted class of graphs.  For instance, for perfect graphs $\chi$ and $\chi_f$ are equal and can be computed efficiently.  In fact, finding an optimal fractional colouring is the key to efficiently $\chi$-colouring a perfect graph \cite{reed01b}; no efficient combinatorial algorithm is known.

Another class for which the fractional approach is fruitful is the class of line graphs (we define them later in this section).  Although computing the chromatic number of a line graph is NP-hard \cite{holyer81}, we can compute the fractional chromatic number of a line graph in polynomial time.  Combining this fact with other results gives us an asymptotic approximation algorithm for the chromatic number of a line graph.  Kahn proved that for line graphs $\chi$ and $\chi_f$ agree asymptotically \cite{kahn96}, and Sanders and Steurer gave an algorithmic result \cite{sanderss05}:

\begin{theorem}\label{thm:ss}
Any line graph $G$ can be coloured in polynomial time using $\chi_f(G) + \sqrt{\frac 92 \chi_f(G)}$ colours.
\end{theorem}

Our aim is to extend this approach to give a polynomial-time algorithm for $(1+o(1))\chi$-colouring a broader class of graphs.  While we can fractionally colour any claw-free graph in polynomial time (as we will explain in Section \ref{sec:algorithmic}), the approach will not give us a polytime method for $(1+o(1))\chi$-colouring claw-free graphs since $\chi$ and $\chi_f$ do not necessarily agree asymptotically.  For example, the complement of a collection of $k$ disjoint 5-cycles has fractional chromatic number $\tfrac 52 k$ and chromatic number $3k$ for any $k\geq 0$.

Thus we consider the class of {\em quasi-line} graphs, a subclass of claw-free graphs that includes all line graphs.  A graph is quasi-line if the neighbours of any vertex can be covered by two cliques.  Chudnovsky and Seymour's recent structure theorem for quasi-line graphs \cite{cssurvey} has sparked interest in bounding the chromatic number of these graphs \cite{chudnovskyo07, chudnovskyo08, kingr08}.  Let $t(G)$ denote $\lfloor \chi_f(G)+3\sqrt{\chi_f(G)} \rfloor$.  Here we prove:

\begin{theorem}\label{thm:main}
Any quasi-line graph $G$ satisfies $\chi(G) \leq t(G)$.
\end{theorem}

Furthermore our proof leads yields an efficient algorithm:

\begin{theorem}\label{thm:mainalgo}
Any quasi-line graph $G$ can be coloured in polynomial time using $t(G)$ colours.
\end{theorem}


Chudnovsky and Seymour's structure theorem states, roughly speaking, that a quasi-line graph $G$ is a circular interval graph, or can be constructed from a line graph by replacing each vertex with a linear interval graph, or contains a structural fault which points out a proper quasi-line subgraph $G'$ with $\chi(G)=\chi(G')$ and $\chi_f(G)=\chi_f(G')$.  This allows us to prove our results by combining results on line graphs, circular interval graphs, and linear interval graphs.  

In Section \ref{sec:sketch} we sketch our proof of Theorem \ref{thm:main}, dealing with the details in subsequent sections and finally proving Theorem \ref{thm:mainalgo} in Section \ref{sec:algorithmic}.  In the remainder of this section we present Chudnovsky and Seymour's characterization of quasi-line graphs.


\subsection{The structural foundation}

In this paper we allow multigraphs to have loops.  Given a multigraph $H$, its {\em line graph} $L(H)$ is the graph with one vertex for each edge of $H$, in which two vertices are adjacent precisely if their corresponding edges in $H$ share at least one endpoint.  We say that $G$ is a line graph if $G=L(H)$ for some multigraph $H$.  Thus the neighbours of any vertex $v$ in a line graph $L(H)$ are covered by two cliques, one for each endpoint of the edge in $H$ corresponding to $v$.

\subsubsection{Linear and circular interval graphs}

A {\em linear interval graph} is a graph $G=(V,E)$ with a {\em linear interval representation}, which is a point on the real line for each vertex and a set of intervals, such that vertices $u$ and $v$ are adjacent in $G$ precisely if there is an interval containing both corresponding points on the real line.  If $X$ and $Y$ are specified cliques in $G$ consisting of the $|X|$ leftmost and $|Y|$ rightmost vertices of $G$ respectively, we say that $X$ and $Y$ are {\em end-cliques} of $G$.

In the same vein, a {\em circular interval representation} of $G$ consists of $|V|$ points on the unit circle and a set of intervals (arcs) on the unit circle such that two vertices of $G$ are adjacent precisely if some arc contains both corresponding points.  The class of {\em circular interval graphs}, i.e.\ those graphs with a circular interval representation, are the first of two fundamental types of quasi-line graph.  Deng, Hell, and Huang proved that we can identify and find a representation of a circular or linear interval graph in linear time \cite{denghh96}.  We now describe the second fundamental type of quasi-line graph.

\subsubsection{Compositions of linear interval strips}

A {\em linear interval strip} $(S,X,Y)$ is a linear interval graph $S$ with specified end-cliques $X$ and $Y$.  To make a composition of linear interval strips, we begin with an underlying directed multigraph $H$ and a corresponding set of disjoint linear interval strips $\{(S_e,X_e,Y_e) \mid e\in E(H)\}$.  For each vertex $v\in V(H)$ we define the {\em hub clique} as
$$C_v = \left( \bigcup\{X_e \mid e \textrm{ is an edge out of } v \}\right) \cup \left( \bigcup \{Y_e \mid e \textrm{ is an edge into } v \}\right).$$
We now construct $G$, the composition of strips $\{(S_e,X_e,Y_e)\mid e\in E(H)\}$ with underlying multigraph $H$, by taking the disjoint union of strips and adding edges to make each $C_v$ a clique.  We say that $G$ is a {\em composition of linear interval strips} (see Figure \ref{fig:strip}).  Let $G_h$ denote the subgraph of $G$ induced on the union of all hub cliques.  That is,
$$G_h = G[\cup_{v\in V(H)} C_v] = G[\cup_{e\in E(H)} (X_e\cup Y_e)].$$

\begin{figure}
\begin{center}
\includegraphics[width=\textwidth]{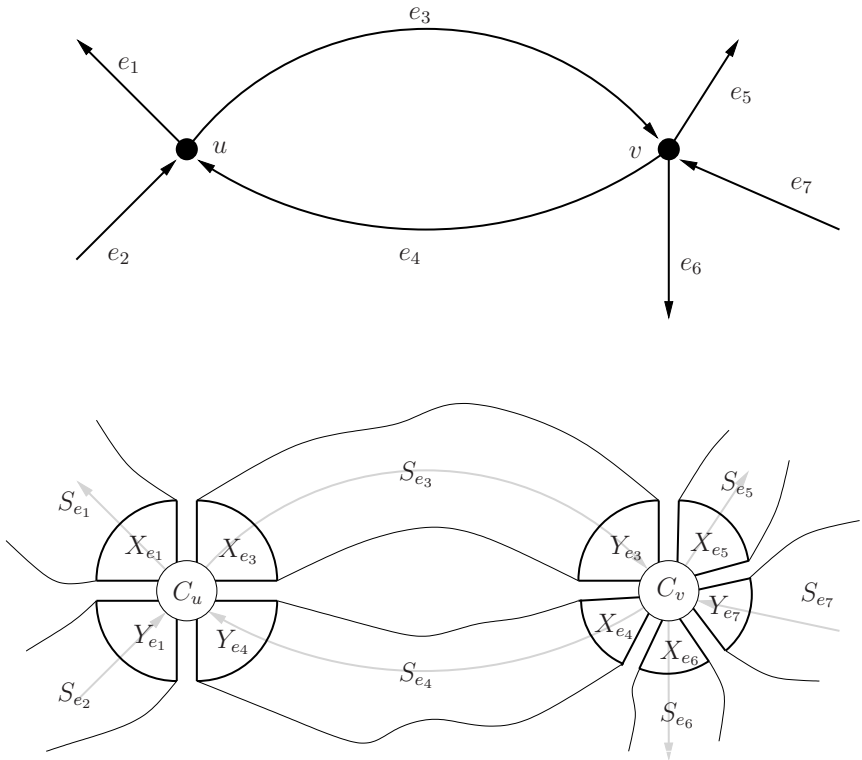}
\end{center}
\caption[A composition of strips]{We compose a set of strips $\{(S_e,X_e,Y_e) \mid e\in E(H)\}$ by joining them together on their end-cliques.  A hub clique $C_u$ will arise for each vertex $u \in V(H)$.}
\label{fig:strip}
\end{figure}

Compositions of linear interval strips generalize line graphs: note that if each $S_e$ satisfies $|S_e|=|X_e|=|Y_e|=1$ then $G = G_h = L(H)$.

\subsubsection{Homogeneous pairs of cliques}

A pair of disjoint nonempty cliques $(A,B)$ is a {\em homogeneous pair of cliques} if $|A\cup B|\geq 3$, and every vertex outside $A\cup B$ sees either all or none of $A$ and either all or none of $B$.  These are a special case of {\em homogeneous pairs}, which were introduced by Chv\'atal and Sbihi in the study of perfect graphs \cite{chvatals87}.  It is not hard to show that for a homogeneous pair of cliques $(A,B)$, $G[A\cup B]$ contains an induced copy of $C_4$ precisely if $A\cup B$ does not induce a linear interval graph; in this case we say that $(A,B)$ is a {\em nonlinear} homogeneous pair of cliques\footnote{These were originally called {\em nontrivial} homogeneous pairs of cliques by Chudnovsky and Seymour, who used them in their description of quasi-line graphs \cite{cssurvey}.  We prefer the more descriptive term {\em nonlinear} in part because {\em linear} homogeneous pairs of cliques are less trivial than {\em skeletal} homogeneous pairs of cliques, which are useful in the study of general claw-free graphs (see \cite{kingthesis}, Chapter 6).}.

\subsubsection{The structure theorem}

Chudnovsky and Seymour's structure theorem for quasi-line graphs \cite{cssurvey} tells us that all quasi-line graphs are made from the building blocks we just described.

\begin{theorem}\label{thm:structure1}
Any connected quasi-line graph containing no nonlinear homogeneous pair of cliques is either a circular interval graph or a composition of linear interval strips.
\end{theorem}

In fact, their analysis of {\em unbreakable quasi-line stripes} in \cite{clawfree7} implies that we need only consider non-degenerate linear interval strips.  We say that a linear interval strip $(S_e,X_e,Y_e)$ is {\em trivial} if $S_e$ consists of a single vertex contained in both $X_e$ and $Y_e$.  We say a linear interval strip $(S_e,X_e,Y_e)$ is {\em canonical} if it is trivial, or if $X_e$ and $Y_e$ are disjoint and nonempty, and $S_e$ is connected.  Theorem 1.1 in \cite{clawfree7} and Theorem 9.1 in \cite{clawfree5} directly imply the desired structure theorem:

\begin{theorem}\label{thm:structure2}
Any connected quasi-line graph containing no nonlinear homogeneous pair of cliques is either a circular interval graph or a composition of canonical linear interval strips.
\end{theorem}

It is actually straightforward (but tedious) to prove that every composition of linear interval strips is a composition of canonical linear interval strips.

As we will show, we can restrict our attention to connected quasi-line graphs with minimum degree at least $t(G)$, containing no nonlinear homogeneous pair of cliques, and containing no clique cutset.  We call such graphs {\em robust}, and note that the structure theorem still applies, and hence:

\begin{theorem}\label{thm:structure}
Any robust quasi-line graph is either a circular interval graph or a composition of canonical linear interval strips.
\end{theorem}

The fact that a robust graph does not contain a clique cutset easily implies that if it is a composition of strips and is not a circular interval graph, the underlying multigraph contains no loops.

\section{A proof sketch}\label{sec:sketch}

Here we sketch the proof of Theorem \ref{thm:main}.  We will show that a minimum counterexample cannot be a circular interval graph, must be robust, and cannot be a composition of canonical linear interval strips.  Thus, Theorem \ref{thm:structure} tells us that Theorem \ref{thm:main} holds.  The fact that no minimum counterexample is a circular interval graph or non-robust follows easily from known results, as we discuss at the end of this section.

It remains to prove that no minimum counterexample to Theorem \ref{thm:main} can be a robust composition of  canonical linear interval strips.  We actually prove directly that every robust composition of canonical linear interval strips has a $t(G)$-colouring.  To do so we take a $t(G)$-colouring $C_h$ of the graph $G_h$ formed by the union of the hub cliques, and combine it with a $t(G)$-colouring $C_e$ of each strip $S_e$.  Since the only edges from the strips to the rest of the graph are from the end-cliques, this will be possible if each $C_e$ agrees with $C_h$ on the end-cliques. Actually, since all the vertices in a given end-clique have identical neighbourhoods outside of the strip that contains them, a weaker condition ensures we can combine them.

\begin{observation}
If there is a $t(G)$-colouring $C_h$ of $G_h$, and each strip $S_e$ has a $t(G)$-colouring $C_e$ such that the following three invariants of $C_e$ and $C_h$ agree, then there is a $t(G)$-colouring of $G$: \begin{itemize}
\item the number of colour classes intersecting both $X_e$ and $Y_e$,
\item the number of colour classes intersecting $X_e$ but not $Y_e$, and
\item the number of colour classes intersecting $Y_e$ but not $X_e$.
\end{itemize}
\end{observation}

\begin{proof}For each $e$, permute the colour class names in $C_e$ so that each stable set is assigned the same colour as a colour class of the same ``type'' in $C_h$, i.e.\ with the same size intersection in both $X_e$ and $Y_e$.  The union of the colourings $C_e$ on the strips is a $t(G)$-colouring of $G$.
\end{proof}

\noindent{\bf Remark:} Since $X_e$ and $Y_e$ are cliques, if $w_e$ is the number of colours appearing in both end-cliques, then there are $|X_e|-w_e$ colours intersecting only $X_e$, and $|Y_e|-w_e$ colour classes intersecting only $Y_e$.  Thus the condition in the observation is simply that for every $e\in E(H)$, $w_e$ is the same for $C_h$ and $C_e$.  In the context of a fractional colouring, $w_e$ denotes the weight of colour classes intersecting both $X_e$ and $Y_e$.\\

Let $G$ be a robust composition of canonical linear interval strips.  Given the observation, Theorem \ref{thm:main} follows from the following three lemmas.  We define
$$t'(G) = \chi_f(G)+\tfrac 1 3\sqrt{\omega(G)}.$$
Note that $\chi_f(G) < t'(G) < t(G)$.

\begin{lemma}\label{lem:a}
Let $G$ be a robust composition of canonical linear interval strips.  There is a fractional $t'(G)$-colouring of $G$ such that $w_e$ is integral for every $e \in E(H)$.  In particular, given any fractional $\chi_f(G)$-colouring of $G$ with overlaps $\{w_e \mid e\in E(H)\}$, there is a fractional $t'(G)$-colouring of $G$ with overlaps $\{\lfloor w_e \rfloor \mid e\in E(H)\}$. 
\end{lemma}

A colouring guaranteed by Lemma \ref{lem:a} gives us fractional colourings of the hub graph and strips on which the $w_e$ agree.  We want to convert these to integral colourings.  To deal with the strips, we prove:


\begin{lemma}\label{lem:b}
For integers $r$ and $k$, suppose a linear interval strip $S_e$ has a fractional $k$-colouring in which $w_e=r$.  Then $S_e$ has an integral $k$-colouring in which $w_e=r$.
\end{lemma}

And to deal with the hub graph, we prove:

\begin{lemma}\label{lem:c}
Let $G$ be a robust composition of canonical linear interval strips.  If $G_h$ has a fractional $t'(G)$-colouring in which each $w_e$ is an integer $r_e$ then it has an integer $t(G)$-colouring in which $w_e=r_e$ for all $e$.
\end{lemma}

We now briefly sketch the proofs of these three lemmas.

To prove Lemma \ref{lem:a} we massage a $\chi_f(G)$-colouring of $G$ to obtain the desired $t'(G)$-colouring.  In doing so we are permitted to use $t'-\chi_f(G) = \frac 13 \sqrt{\omega(G)}$ new colours, and we can essentially handle each strip separately.  If a strip $(S_e,X_e,Y_e)$ is trivial then we say that $e$ is a {\em trivial edge}.  If $e$ is trivial, note that in any fractional colouring, $w_e=|X_e|=|Y_e|=1$, thus $w_e$ is an integer and no massaging is necessary.  The difficulty in proving Lemma \ref{lem:a} lies in handling the nontrivial strips.
We will see that for any nontrivial strip $e$ we can modify any fractional colouring of $G$, at the cost of at most one extra colour,  so as to make the overlap $w_e$ integral. Furthermore, it turns out that for any set of nontrivial strips, if we choose an endpoint of each corresponding edge so that no vertex of $H$ is chosen twice,  we can use the same extra colour for all these strips, and therefore make all their overlaps integral at the cost of one extra colour.  As we see below, it follows immediately from the fact that $G$ is robust and hence has minimum degree at least $t(G)$, that we can partition the nontrivial strips into $\frac 13 \sqrt{\omega(G)}$ such sets.  We flesh out this proof of Lemma \ref{lem:a} in the next section.

Lemma \ref{lem:b} follows naturally from the fact that for circular interval graphs, $\chi=\lceil \chi_f\rceil$ (see Lemma \ref{lem:nk} below).  We give the details in Section \ref{sec:emulating}.

To prove Lemma \ref{lem:c} we exploit the fact that $G_h$ closely resembles a line graph $G'$ constructed as the composition of strips $\{ (S'_e,X'_e,Y'_e) \mid e\in E(H)\}$ such that
\begin{itemize}
\item $S'_e$ is a clique of size $|X_e|+|Y_e|-w_e$.
\item $X'_e$ and $Y'_e$ are cliques of size $|X_e|$ and $|Y_e|$ respectively, with $|X'_e\cap Y'_e|=w_e$.
\end{itemize}
\begin{figure}
\includegraphics[scale=.9]{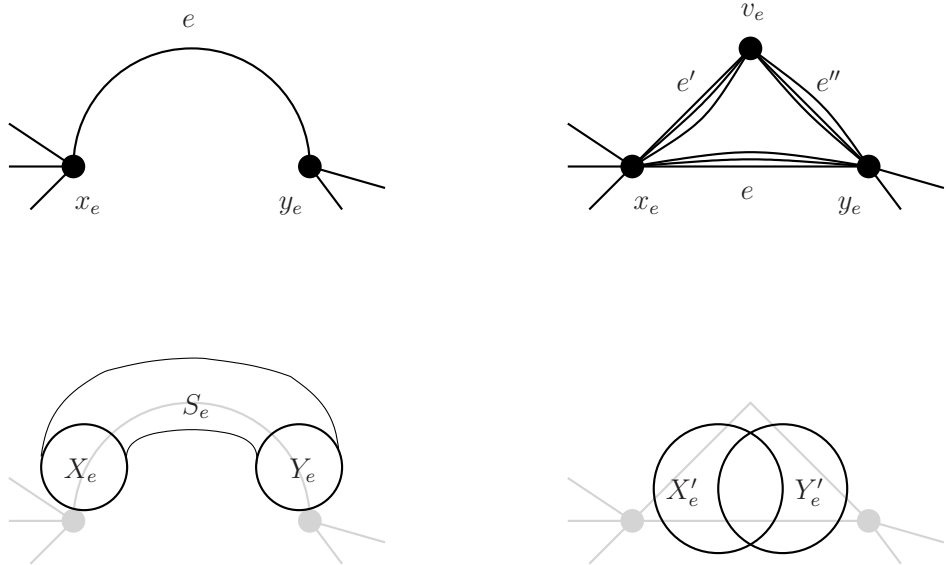}
\caption{We construct a line graph $G'$ from $G$ by contracting each strip $(S_e,X_e,Y_e)$ into a $(S'_e,X'_e,Y'_e)$ where $S'_e$ is a clique, $X'_e \cup Y'_e$ together cover $S'_e$, and $|X'_e\cap Y'_e|=w_e$.}
\label{fig:contraction}
\end{figure}
Now $G'$ is the line graph of the multigraph $H'$, which we construct from $H$ as follows.  Take each $e\in E(H)$, say from $x_e$ to $y_e$.  Add a vertex $v_e$, put $w_e$ edges between $x_e$ and $y_e$, put $|X_e|-w_e$ edges between $x_e$ and $v_e$, and put $|Y_e|-w_e$ edges between $y_e$ and $v_e$ (see Figure \ref{fig:contraction}).  It is easy to confirm that $G'=L(H')$.  In Section \ref{sec:proof} we prove Lemma \ref{lem:c} by applying Theorem \ref{thm:ss} to $G'$.

\subsection{The easy cases}

Before dealing with compositions of canonical linear interval strips we must explain why a graph $G$ cannot be a minimum counterexample to Theorem \ref{thm:main} or Theorem \ref{thm:mainalgo} if it is a circular interval graph or it is not robust.  This follows immediately from four known lemmas.  The first is a result of Niessen and Kind \cite{niessenk00}:

\begin{lemma}\label{lem:nk}
For any circular interval graph $G$, $\chi(G) = \lceil \chi_f(G)\rceil$.
\end{lemma}

The second, due to Shih and Hsu \cite{shihh89}, tells us that we can optimally colour circular interval graphs efficiently:

\begin{lemma}\label{lem:shihhsu}
Given a circular interval graph $G$, we can find an optimal colouring of $G$ in $O(n^{3/2})$ time.
\end{lemma}

If $G$ contains a vertex $v$ of degree less than $t(G)$, we can find a $t(G)$-colouring of $G-v$ and extend it to a $t(G)$-colouring of $G$ by giving $v$ some colour not appearing in its neighbourhood.  Furthermore, it is well known that we can efficiently decompose and combine colourings on clique cutsets (see e.g.\ \cite{kingthesis}, $\S$3.4.3).  Thus to prove that a minimum counterexample must be robust, we only need to show that a minimum counterexample cannot contain a nonlinear homogeneous pair of cliques.  This is implied by two more known results:

\begin{lemma}[\cite{kingr08}]\label{lem:hpocreduce}
Let $G$ be a quasi-line graph on $n$ vertices containing a nonlinear homogeneous pair of cliques $(A,B)$.  In $O(n^{5/2})$ time we can find a proper quasi-line subgraph $G'$ of $G$ such that $\chi(G)=\chi(G')$, and given a $k$-colouring of $G'$ we can find a $k$-colouring of $G'$ in $O(n^{5/2})$ time.
\end{lemma}

\begin{lemma}[\cite{kingr08}]\label{lem:hpocfind}
Given a quasi-line graph $G$, in $O(n^2m)$ time we can find a nonlinear homogeneous pair of cliques in $G$ or determine that none exists.
\end{lemma}

\section{Finding a good fractional colouring}\label{sec:overlap}

Let $G$ be a robust composition of canonical linear interval strips $\{(S_e,X_e,Y_e) \mid e\in E(H)\}$.  We now prove Lemma \ref{lem:a} by constructing a fractional $t'(G)$-colouring of $G$ with integer overlaps, i.e.\ with $w_e$ integral for all $e\in E(H)$. To find this fractional colouring we take a fractional $\chi_f(G)$-colouring and modify it on the cliques $X_e$.  We can assume that for all $e\in E(G)$, $|X_e|\leq |Y_e|$: If $|X_e|>|Y_e|$, simply change the direction of $e$ in $H$ and swap the names of $X_e$ and $Y_e$.


Our first step in the proof of Lemma \ref{lem:a} is to prove that we don't need to modify too many cliques $X_e$ with edges between them.  For a trivial edge $e$, $w_e$ is always equal to $|X_e|$ and is therefore an integer.  For $v\in V(H)$, let $D(v)$ denote the number of nontrivial edges out of $v$.  Let $D(H)$ be the maximum of $D(v)$ over all $v\in V(H)$.  We bound $D(H)$ using the following result:

\begin{lemma}\label{lem:boundd}
If $e$ is a nontrivial edge of $H$, then $|X_e|\geq 3 \sqrt{\chi_f(v)}$.
\end{lemma}
\begin{proof}
Since $G$ is robust, every vertex in $G$ must have degree at least $\omega(G)+3 \sqrt{\chi_f(G)}-1 \leq t(G)$.

Let the vertices of $S_e$ be $\{u_1,\ldots,u_{V(S_e)}\}$ in linear order, such that $X_e = \{u_1,\ldots,u_{|X_e|}\}$.  Because $e$ is nontrivial and canonical, the vertex $u_{|X_e|+1}$ exists.  By the structure of linear interval graphs, its closed neighbourhood outside $X_e$ is a clique if $u_{|X_e|+1}$ is not in $Y_e$.  If $u_{|X_e|+1}\in Y_e$ then its neighbourhood outside $X_e$ is contained in the hub clique containing $Y_e$.  Therefore $\omega(G)+|X_e| \geq d(u_{|X_e|+1})+1 \geq \omega(G)+ 3 \sqrt{\chi_f(G)}$.
\end{proof}

This immediately gives us a bound on $D(H)$, because every hub clique $C_v$ has size at most $\omega(G)\leq \chi_f(G)$.

\begin{corollary}
$D(H) \leq \frac 1 3\sqrt{\omega(G)} = t'(G)-\chi_f(G)$.
\end{corollary}

We use our bound on $D(H)$ to prove the existence of the desired fractional $t'(G)$-colouring.

\begin{lemma}\label{lem:hedgecolouring}
We can colour the nontrivial edges of $H$ with $D(H)$ colours such that no two edges out of the same vertex get the same colour.
\end{lemma}

To prove this, we simply construct an appropriate colouring by enumerating the nontrivial edges out of each vertex of $H$.  Having done this, we deal with one of these colour classes at a time:

\begin{lemma}\label{lem:overlap1}
Let $E_1$ be a set of nontrivial edges in $H$, no two of which go out of the same vertex.  Given a fractional $k$-colouring of $G$ with overlaps $\{ w_e \mid e \in E(H)\}$, there is a fractional $(k+1)$-colouring of $G$ with overlaps $\{w'_e \mid e \in E(H)\}$ such that $w'_e = \lfloor w_e\rfloor$ for $e\in E_1$, and $w'_e=w_e$ for $e\notin E_1$.
\end{lemma}
\begin{proof}
Take some single edge $e$ of $E_1$ and an optimal fractional colouring of $G$.  We claim that we can make $w'_e$ an integer by adding $w_e-\lfloor w_e \rfloor$ extra weight to the fractional colouring.  To do this, take a collection of stable sets, each intersecting both $X_e$ and $Y_e$, of total weight $w_e-\lfloor w_e \rfloor$ in the colouring (it may be necessary to split one stable set into two identical stable sets of lesser weight to do this).  Now remove the vertex in $X_e$ from each of these stable sets, and fill the missing weight in $X_e$ (i.e.\ $w_e-\lfloor w_e \rfloor$) with singleton stable sets.  This gives us the desired fractional colouring in which $w'_e$ is an integer.  Note that we did not change the colouring outside $X_e$, so every other overlap is unchanged.

To see that we can ensure that every overlap in  $\{w_e \mid e\in E_1\}$ is an integer using extra weight less than 1, note that $\{X_e \mid e\in E_1\}$ is a set of disjoint cliques with no edges between them.  Thus instead of making one $w_e$ an integer by inserting singleton stable sets, we can make every $\{w_e \mid e\in E_1\}$ an integer by inserting stable sets of size $\leq |E_1|$, with total weight less than 1.
\end{proof}

This gives us an easy proof of Lemma \ref{lem:a}:

\begin{proof}[Proof of Lemma \ref{lem:a}]
Begin with a fractional $\chi_f(G)$-colouring of $G$ and a colouring of the edges of $H$ guaranteed by Lemma \ref{lem:hedgecolouring}.  For each matching in this edge colouring, apply Lemma \ref{lem:overlap1}.  Since $\chi_f(G)+D(H) \leq t'(G)$, the result is a fractional $t'(G)$-colouring of $G$ for which each overlap $w_e$ is the round-down of the original overlap.
\end{proof}

\section{Fractional and integer colourings of linear interval strips}\label{sec:emulating}

We now prove Lemma \ref{lem:b}, which tells us that we can emulate fractional colourings of linear interval strips using integer colourings.

\begin{proof}[Proof of Lemma \ref{lem:b}]
Consider a fractional $k$-colouring in which the total weight of colour classes intersecting both $X_e$ and $Y_e$ is an integer $r$.  Let $D$ be the set of colours appearing in both $X_e$ and $Y_e$, let $A$ be the set appearing in $X_e$ but not $Y_e$, $B$ the set appearing in $Y_e$ but not $X_e$, and $C$ the set appearing in neither $X_e$ nor $Y_e$.  For any set $T$ of colours, let $wt(T)$ denote the total weight of the colours in $T$.

Observe that since $k$, $r$, $|X_e|$ and $|Y_e|$ are integers, all of $wt(D)=r$, $wt(A) = |X_e|-r$, $wt(B)= |Y_e|-r$ and $wt(C) = k-|X_e|-|Y_e|+r$ are integers.  We construct a circular interval graph $F_e$ based on the fractional colouring of $S_e$.

Say the $n$ vertices of $S_e$ are $v_1, \ldots, v_n$, left-to-right.  To construct $F_e$ from $S_e$, we first add 
cliques $V_A$, $V_C$, and $V_B$, in order from left to right, to the right of $v_n$.  These cliques have size $wt(A)$, $wt(C)$, and $wt(B)$ respectively.  We then add edges to make three new maximal cliques:
\begin{eqnarray*}
I_X &:=& Y_e\cup V_A \cup V_C\\
I_C &:=& V_A\cup V_C \cup V_B\\
I_Y &:=& V_C\cup V_B \cup X_e
\end{eqnarray*}
Thus $V_A$ is complete to $Y_e$ and $V_C$, $V_C$ is complete to $V_A$, $V_B$, $X_e$, and $Y_e$, and $V_B$ is complete to $V_C$ and $X_e$.  Since $I_X$ and $I_Y$ define $k$-cliques and $I_C$ defines a $k-r$ clique, it is easy to see that $F_e$ is a circular interval graph with clique number $k$, hence both the fractional chromatic number and the chromatic number of $F_e$ are at least $k$.

Now we construct a fractional $k$-colouring of $F_e$.  On the vertices belonging to $S_e$, we keep our initial fractional colouring.  We can then cover $V_A$ with the colours in $A$, $V_C$ with the colours in $C$, and $V_B$ with the colours in $B$.  Hence $\chi_f(F_e) = k$ and so by Lemma \ref{lem:nk} we know that $\chi(F_e)=k$.  So consider an integer $k$-colouring of $F_e$.  We claim that on $S_e$ this is an integer $k$-colouring with exactly $r$ colours that appear in both $X_e$ and $Y_e$.

Suppose fewer than $r$ colours appear in both $X_e$ and $Y_e$.  Then there are more than $|X_e|+|Y_e|-r$ colours that cannot appear in $V_C$.  But $V_C$ is a clique of size $k-(|X_e|+|Y_e|-r)$, contradicting the fact that we have a proper $k$-colouring.  Now suppose more than $r$ colours appear in both $X_e$ and $Y_e$.  Then none of these colours can appear in $V_A \cup V_C \cup V_B$.  But $V_A \cup V_C \cup V_B$ is a clique of size $k-r$ so again we cannot have a proper $k$-colouring.  Thus exactly $r$ colours appear in both $X_e$ and $Y_e$.
\end{proof}

We note that the proof of Lemma \ref{lem:b} actually yields a polynomial-time algorithm which constructs a circular interval graph $F_e$ with chromatic number $k$, such that every optimal colouring of $F_e$ contains a colouring of $S_e$ in which $X_e$ and $Y_e$ have exactly $w_e$ colours in common, as desired.

\section{Completing the proof}\label{sec:proof}

We are now ready to complete the proof of Theorem \ref{thm:main} by proving Lemma \ref{lem:c}, which gives the desired bound on the chromatic number of a robust composition of linear interval strips as explained in Section \ref{sec:sketch}.

The key to the proof is the following:

\begin{lemma}\label{lem:newlemma}
Let $G$ be a robust composition of strips.  Given a fractional $t'(G)$-colouring $C$ of $G$ in which every overlap $w_e$ is an integer, we can construct a line graph $G' = G'(C)$ such that:
\begin{itemize}
\item[(i)] $G'$ has a fractional $t'(G)$-colouring, which implies that $\chi(G')\leq t(G)$, and
\item[(ii)] for any $k\geq t(G)$, any proper $k$-colouring of $G'$ yields a proper $k$-colouring of $G$ such that for each strip $(S_e, X_e,Y_e)$, exactly $w_e$ colours appear in both $X_e$ and $Y_e$.
\end{itemize}
\end{lemma}

Note that Lemma \ref{lem:c} follows immediately from (ii).  Before proving Lemma \ref{lem:newlemma}, we prove the claim that if $\chi_f(G')\leq t'(G)$, then $\chi(G')\leq t(G)$:

\begin{proof}[Proof of Claim]
The desired bound follows from Theorem \ref{thm:ss}:
\begin{eqnarray}
\chi(G') &\leq& \chi_f(G')+ \sqrt{\tfrac 92 \chi_f(G')}\\
&\leq& t'(G)+ \sqrt{\tfrac 92 t'(G)}\\
&\leq &\chi_f(G) + \tfrac 1 3\sqrt{\chi_f(G)} + \sqrt{\tfrac 92 \left(\chi_f(G) + \tfrac 13 \sqrt{\chi_f(G)}\ \right)}\\
&\leq &\chi_f(G) + \tfrac 1 3\sqrt{\chi_f(G)} + \sqrt{\tfrac 92 \left( \tfrac 4 3\chi_f(G) \right)}\\
&\leq &\chi_f(G) + 3\sqrt{\chi_f(G)}. 
\end{eqnarray}
Since $\chi(G')$ is an integer, this implies $\chi(G')\leq t(G)$.
\end{proof}

\begin{proof}[Proof of Lemma \ref{lem:newlemma}]
Suppose $G$ is a composition of strips $\{(S_e,X_e,Y_e)\mid e\in E(H)\}$ with underlying multigraph $H$.  We construct $G'$ as the composition of new strips $\{(S'_e,X'_e,Y'_e)\mid e\in E(H)\}$ with underlying multigraph $H$ as follows.  We replace every strip $(S_e,X_e,Y_e)$ with a new strip $(S'_e,X'_e,Y'_e)$ such that
\begin{itemize}
\item[-] $S'_e$ is a clique of size $|X_e|+|Y_e|-w_e$.
\item[-] $X'_e$ and $Y'_e$ have size $|X_e|$ and $|Y_e|$ respectively and cover the vertices of $S'_e$.  This implies that $X'_e\cap Y'_e$ has size $w_e$.
\end{itemize}
As explained at the end of Section \ref{sec:sketch}, $G'$ is a line graph (see Figure \ref{fig:contraction}).

We can see that $\chi_f(G') \leq t'(G)$: Given our fractional colouring $C$ of $G$, we simply cover each $X'_e$ (resp.\ $Y'_e$) with the colours appearing in $X_e$ (resp.\ $Y'_e$).  We can do this since $S'_e$ is a clique and $|X'_e\cap Y'_e|=w_e$.  And since every vertex in $X'_e\setminus Y'_e$ (resp.\ $Y'_e\setminus X'_e$, $X'_e\cap Y'_e$) has the same neighbourhood in $G'$, each resulting colour class is a stable set.

Now consider an integer colouring of $G'$ using $t(G)$ colours, which we know exists due to the previous claim.  Since $S'_e$ is a clique for every edge $e$ of $H$, the number of colours appearing in both $X'_e$ and $Y'_e$ is precisely $w_e$.  Note that in $G_h$, any vertex in $X_e$ (resp.\ $Y_e$) has the same neighbourhood outside $X_e\cup Y_e$.  Furthermore, $X_e\cup Y_e$ is the complement of a bipartite graph containing a matching of size $w_e$.  Therefore from our integer colouring of $G'$ we can construct an integer colouring of $G_h$ in which precisely $w_e$ colours appear in both $X_e$ and $Y_e$ for every $e\in E(H)$.
\end{proof}

\section{Algorithmic considerations}\label{sec:algorithmic}

We now prove Theorem \ref{thm:mainalgo}, which states that we can $t(G)$-colour a quasi-line graph $G$ in polynomial time.  We wish to reduce the problem to that of colouring robust compositions of linear interval strips.  To do so, we proceed as follows:

\begin{enumerate}
\item[(i)]If there is a vertex $v$ of degree less than $t(G)$, remove it, recursively $t(G)$-colour $G-v$, and give $v$ a colour not appearing in its neighbourhood.
\item[(ii)]If possible, find a nonlinear homogeneous pair of cliques $(A,B)$ using Lemma \ref{lem:hpocfind}, reduce it using Lemma \ref{lem:hpocreduce}, and use the colouring of the reduced graph to find a $t(G)$-colouring of $G$, again using Lemma \ref{lem:hpocreduce}.
\item[(iii)]If $G$ contains a clique cutset, we decompose $G$ on the clique cutset, colour the resulting graphs, and combine the colourings on the clique cutset, as described in Section 3.4.3 of \cite{kingthesis}.
\item[(iv)]If we reach this case then $G$ is robust; we now appeal to our algorithm for $t(G)$-colouring a robust quasi-line graph.
\end{enumerate}

To colour a robust composition of canonical linear interval strips $G$, we use the following steps:

\begin{enumerate}
\item Find a canonical linear interval strip decomposition for $G$, i.e.\ a multigraph $H$ and canonical linear interval strips $\{(S_e,X_e,Y_e) \mid e\in E(H)\}$.

\item Find an optimal fractional colouring of $G$, and compute $\lfloor w_e\rfloor$ for each nontrivial edge $e$ in $E(H)$.

\item Using the strip decomposition of $G$ and the values $\{\lfloor w_e\rfloor \mid e\in E(H)\}$, construct $G'$.

\item Construct an integer $t(G)$-colouring of $G'$, as guaranteed by Theorem \ref{thm:ss}.  From it, construct a $t(G)$ colouring of $G_h$ by finding a matching of size $\lfloor w_e \rfloor$ in the bipartite graph $\overline G [X_e\cup Y_e]$ for each nontrivial edge $e$.  This gives us a partial $t(G)$-colouring of $G$ in which every overlap $\lfloor w_e\rfloor$ is the same as in Step 2.
\item For each strip $(S_e,X_e,Y_e)$ in the decomposition of $G$, construct an integer $t(G)$-colouring of $S_e$ with the same overlap $\lfloor w_e\rfloor$ appearing in Step 2.
\item Combine these strip colourings with the colouring of $G_h$ to reach a proper $t(G)$-colouring of $G$.
\end{enumerate}

We can easily perform Steps 3 and 4 in polynomial time.  As we mentioned in Section \ref{sec:emulating}, to find the desired colouring of $S_e$ we must first construct an auxiliary circular interval graph $F_e$ as in the proof of Lemma \ref{lem:b}, then find an optimal colouring of it and restrict the colouring to the vertices of $S_e$.  It is easy to confirm that given $w_e$ we can construct $F_e$ in polynomial time; to find an optimal colouring of $F_e$ in polynomial time we appeal to Lemma \ref{lem:shihhsu}.  To complete the colouring at Step 6, we begin with our colouring of $G_h$.  Then for every nontrivial strip $S_e$, we take the colour classes in the colouring of $S_e$ and reassign them colours from the colouring of $G_h$ based on whether they intersect $X_e$, $Y_e$, both, or neither.  We can freely permute colours in $X_e$ and $Y_e$ because both $X_e$ and $Y_e$ are homogeneous cliques in $G_h$.  Clearly we can do this in polynomial time.  Thus it only remains for us to prove that we can perform Steps 1 and 2 in polynomial time.  We spend the remainder of the section on this task.

\subsection{Decomposing a composition of canonical linear interval strips}

Given a graph $G$ which is a composition of canonical linear interval strips, we must decompose $G$ into canonical linear interval strips and an underlying multigraph $H$.  Algorithms for doing this in polynomial time have been provided independently by Faenza, Oriolo, and Stauffer \cite{faenzaos11}, Hermelin, Mnich, van Leeuwen, and Woeginger \cite{hermelinmlw11}, and Chudnovsky and King \cite{chudnovskyk11}.  This implies that we can perform Step 1 of our algorithm in polynomial time.  It only remains to deal with Step 2.

\subsection{Constructing the line graph}

To achieve Step 2, we must first find an optimal fractional colouring of $G$.  We then calculate each overlap $w_e$, which gives us each $\lfloor w_e \rfloor$ that we need to construct the line graph $G'$.

Minty \cite{minty80}, Nakamura and Tamura \cite{nakamurat01}, and recently Oriolo, Pietropaoli, and Stauffer \cite{oriolops08} give polynomial-time algorithms for finding a maximum-weight stable set in any claw-free graph.  By polynomial equivalence results of Gr\"otschel, Lov\'asz, and Schrijver (see \cite{grotschells81} \S 7), this implies that we can compute the fractional chromatic number of any claw-free graph in polynomial time.  We can deal with weights as well:  Given positive integer weights $\beta(v)$ for each vertex $v$ of the graph $G$, let $G_\beta$ be the graph obtained from $G$ by substituting a clique of size $\beta(v)$ for each vertex $v$ of $G$.  We can find $\chi_f(G_\beta)$ in polynomial time in terms of $|V(G)|+\log(\max_{v\in V(G)}\beta(v))$.

But in fact more is true.  The equivalence results in \cite{grotschells81} tell us that for claw-free graphs, we can efficiently optimize over the fractional clique polytope.  This is the dual feasible region of the linear program describing the fractional chromatic number.  Theorem 6.5.14 in \cite{grotschellsbook} states:

\begin{theorem}
There exists an oracle-polynomial time algorithm that, for any well-described polyhedron $(P; n, \phi)$ given by a strong separation oracle and for any $c\in \mathbb Q^n$, either
\begin{enumerate}
\item finds a basic optimum standard solution, or
\item asserts that the dual problem is unbounded or has no solution.
\end{enumerate}
\end{theorem}

We refer the reader to \cite{grotschellsbook} for the formal definitions of a {\em well-described polyhedron} and an {\em oracle-polynomial time algorithm}.  In our specific situation, this theorem states that since we can optimize over the fractional clique polytope, we can find a {\em basic solution} to the fractional chromatic number linear program in polynomial time for any claw-free graph.  A basic solution to this linear program is a fractional colouring in which the incidence vectors of the stable sets with nonzero weight are linearly independent.  Thus by the dimension of the stable set polytope, there are at most $|V(G)|$ nonzero stable sets in the fractional colouring.  Therefore not only can we construct an optimal fractional colouring of $G$ in polynomial time, but we can easily compute $w_e$ for each strip $(S_e,X_e,Y_e)$ in $G$.

This is the final piece of the algorithmic puzzle, and completes the proof of Theorem \ref{thm:mainalgo}.

\section{Conclusion}

We have proven that we can $(1+o(1))\chi_f$-colour any quasi-line graph in polynomial time.  While we cannot do this for all claw-free graphs, in a later paper we will prove that we can do it whenever the stability number is at least four.  The approach is essentially the same but involves a lot of case analysis.

We conjecture that an asymptotic approximation algorithm exists for claw-free graphs even when $\chi$ and $\chi_f$ do not agree asymptotically:

\begin{conjecture}
There is an algorithm that, given any claw-free graph $G$, returns a proper $(1+o(1))\chi(G)$-colouring of $G$ in polynomial time.
\end{conjecture}

\section{Acknowledgements}

We would like to thank Bruce Shepherd for helpful conversations about the work in this paper, and Anna Galluccio for providing valuable feedback on an earlier version of this work.



\end{document}